%% file: paper.tex
\documentclass{eptcs}

\input{macros}

\begin{document}
\title{An Automata Theoretic Approach to the Zero-One Law for Regular Languages:
Algorithmic and Logical Aspects}
\author{Ryoma Sin'ya
\institute{Tokyo Institute of Technology.}
\email{shinya.r.aa@m.titech.ac.jp}
\institute{\'Ecole Nationale Sup\'erieure des T\'el\'ecommunications.}
\email{rshinya@enst.fr}}
\def\titlerunning{{An Automata Theoretic Approach to the Zero-One Law for Regular Languages}}
\def\authorrunning{R. Sin'ya}

\maketitle
%\keywords{zeor-one laws, syntactic monoids with zero,
%zero automata, finite model theory, equational theory,
%synchronisation words}

\input{chaps/abstract.tex}
\input{chaps/introduction.tex}
\input{chaps/preliminaries.tex}
\input{chaps/synchronisation.tex}
\input{chaps/bqa.tex}
\input{chaps/zero.tex}
\input{chaps/algorithm.tex}
\input{chaps/logics.tex}
\input{chaps/conclusion.tex}
\newpage
\bibliographystyle{eptcs}
\bibliography{ref}

\end{document}

%% file: macros.tex
\usepackage[dvipdfm]{graphicx}
\usepackage{url}
\usepackage{bm}
\usepackage{amsmath,amssymb,amsthm}
\usepackage{float}
\usepackage{color}
\usepackage{extarrows}
\usepackage{enumerate}

\def\CA{{\cal A}}
\def\CB{{\cal B}}

\def\order{\mathop{\mathrm{O}}}
\def\Z{{\cal Z}}
\def\ZO{{\cal Z\!\!O}}

\def\Q{\!/\!\!\!\sim}

\def\past{{\sf{Past}}}
\def\future{{\sf{Fut}}}
\def\sink{{\sf{Sink}}}

\makeatletter
  \def\@thmcountersep{.}
\makeatother
\newtheorem{theorem}{Theorem}

\newtheorem{corollary}{Corollary}
\newtheorem{lemma}{Lemma}
\newtheorem{example}{Example}
\newtheorem{remark}{Remark}
\newtheorem{definition}{Definition}
\newtheorem{proposition}{Proposition}

\newcommand{\FO}{\mathop{\mathrm{FO}}\nolimits}

\newcommand{\MSO}{\mathop{\mathrm{MSO}}\nolimits}

\newcommand{\bool}{\mathbb{B}}

\newcommand{\limn}{\lim_{n \rightarrow \infty}}

\newcommand{\ie}{i.e., \;}
\newcommand{\eg}{e.g., \;}
\newcommand{\cf}{{\it cf.\;}}

\newcommand{\sect}[1]{\noindent{\bf #1.\;}}
\newcommand{\con}[1]{\textcolor{red}{\textcircled{\scriptsize\ref{condition:#1}}}}

\makeatletter
\def\Hline{%
\noalign{\ifnum0=`}\fi\hrule \@height 1pt \futurelet
\reserved@a\@xhline}
\makeatletter
\def\iroha#1{%
  \ifcase#1\or い\or ろ\or は\or に\else\@ctrerr\fi}
\makeatother
\def\maru#1{\textcircled{\scriptsize#1}}

%% file: chaps/abstract.tex
\begin{abstract}
 A zero-one language $L$ is a regular language whose asymptotic probability
 converges to either zero or one. In this case, we say that $L$ obeys
 the zero-one law.
 We prove that a regular language obeys the zero-one law if and only if its
 syntactic monoid has a zero element, by means of Eilenberg's
 variety theoretic approach.
 Our proof gives an effective automata characterisation of the zero-one
 law for regular languages, and it leads to a linear time algorithm for
 testing whether a given regular language is zero-one. 
 In addition, we discuss the logical aspects of the zero-one law for
 regular languages.
\end{abstract}

%% file: chaps/introduction.tex
\section{Introduction}\label{introduction}
Let $L$ be a regular language over a non-empty finite alphabet $A$.
Recall that the {\it counting function} $\gamma_n(L)$ of $L$ counts
the number of different words of length $n$ in $L$: $\gamma_n(L) = |L
\cap A^n|$ where $A^n$ is the set of all words of length $n$ over $A$.
The {\it probability function} $\mu_n(L)$ of $L$ is the fraction
defined by
\[
 \mu_n(L) = \frac{\gamma_n(L)}{\gamma_n(A^*)} = \frac{|L \cap A^n|}{|A^n|}.
\]
The {\it asymptotic probability} $\mu(L)$ of $L$ is defined by $\mu(L) =
\limn \mu_n(L)$, if the limit exists. 
We can regard $\mu_n(L)$ as the {\it probability} that a randomly chosen
word of length $n$ is in $L$, and $\mu(L)$ as its {\it asymptotic
probability}. Here we introduce a new class of regular languages which
is the main target of this paper.

\begin{definition}[zero-one language]\upshape
A {\it zero-one language} $L$ is a regular language whose asymptotic probability
 $\mu(L)$ is either zero or one. In this case, we say that $L$ {\it
 obeys the zero-one law}. We denote by $\ZO$ {\it the class of all
 regular zero-one languages}.
\end{definition}

As we will describe later (see Section \ref{logics}), the notion of ``zero-one
law'' defined here is a fundamental object in {\it finite model theory}. 

\begin{example}\label{ex:probability}\upshape
We now consider a few examples.
\begin{itemize}
 \item The set of all words $A^*$ over $A$ satisfies $\mu(A^*) = 1$, and
	   its complement $\emptyset$ satisfies $\mu(\emptyset) = 0$. These
	   two languages obey the zero-one law.
 \item Consider $aA^*$ the set of all words which start with the letter
	   $a$ in $A$. Then 
	   \[
		\mu_n(aA^*) = \frac{|aA^{n-1}|}{|A^n|} = \frac{1}{|A|}.
	   \]
	   Hence, its limit $\mu((aA)^*)$ is $1/|A|$ and $aA^*$ is zero-one
	   if and only if $A$ is {\it unary}: $A = \{a\}$.
 \item Consider $(AA)^*$ the set of all words with even length. Then
	   \[
		\mu_n((AA)^*) = \begin{cases}
						 1 & \text{if} \;\; n \;\; \text{is even,}\\
						 0 & \text{if} \;\; n \;\; \text{is odd.}
						\end{cases}
	   \]
	   Hence, its limit $\mu((AA)^*)$ does not exist. 
\end{itemize}
\end{example}
Thus, for some regular language $L$, the asymptotic probability $\mu(L)$ is
either zero or one, for some, like $L = aA^*$ where $|A| \geq 2$, $\mu(L)$ could
be a real number between zero and one, and for some, like $L =
(AA)^*$, it may not even exist. It is previously known that there exists a cubic
time algorithm computing $\mu(L)$ for any regular language $L$
(\cite{Bodirsky}, see Section \ref{conclusion}).\\

\sect{Our results and contributions}
In this paper, we show that the following class of languages exactly
captures the zero-one law for regular languages.
\begin{definition}[\cite{MPRI}]\upshape
A {\it language with zero} is a regular language whose syntactic monoid
 has a zero element.
 We denote by $\Z$ {\it the class of all regular languages with zero}.
\end{definition}

More precisely, we prove the following theorem, which states that $\ZO$
and $\Z$ are equivalent by means of a transparent condition of their
automata: {\it zero automata} (Section \ref{synchronisation}) and {\it
quasi-zero automata} (Section \ref{algorithm}) which will be described
later.
The remarkable fact is that, $\ZO = \Z$ holds even though these two
notions seem completely different from each other; $\ZO$ is defined by
the asymptotic behavior of its probability, $\Z$ is defined by the
existence of a zero of its syntactic monoid.

\begin{theorem}\label{thm:zero}
Let $L$ be a regular language and $\CA_L$ be the minimal automaton of L.
Then the following four conditions are equivalent.
\begin{enumerate}
\renewcommand{\labelenumi}{\maru{\arabic{enumi}}}
 \item $\CA_L$ is zero. \label{condition:za}
 \item $L$ is with zero. \label{condition:zero}
 \item $L$ obeys the zero-one law. \label{condition:zeroone}
 \item $L$ is recognised by a quasi-zero automaton. \label{condition:quasi}
\end{enumerate}
\end{theorem}

We will prove this theorem as a cyclic chain of implications:
 $\con{za} \Rightarrow \con{zero} \Rightarrow \con{zeroone} \Rightarrow
 \con{za}$, and $\con{za} \Leftrightarrow \con{quasi}$ independently.
We should notice that the most difficult part of this proof is the
 implication $\con{zeroone} \Rightarrow \con{za}$,
 while the former part $\con{za} \Rightarrow \con{zero}
 \Rightarrow \con{zeroone} $ is easy.
The key points of the proof of this part are {\it closure properties of
 $\ZO$} and Lemma \ref{lemma:finalstates}, which comes from Eilenberg's
 variety theorem.
The automata characterisation $\con{quasi}$ of Theorem \ref{thm:zero} leads to a linear
 time algorithm for testing whether a given regular language is
 zero-one. In addition, our automata theoretic proof sheds new light on
 the relation between the zero-one law for regular languages and {\it
 logical fragments over finite words}.\\

\sect{Paper outline}
The remainder of this paper is organised as follows.
In Section \ref{preliminaries}, we first give the necessary definitions
and terminology for languages, monoids, and automata.
Lemma \ref{lemma:finalstates} will be introduced in this section.
For the sake of completeness we include the proof of Lemma
\ref{lemma:finalstates}.
Section \ref{synchronisation} provides a detailed exposition of the
notion of zero automata.
Our automata theoretic proof of Theorem \ref{thm:zero} consists of three parts:
 (i) Check certain closure properties of $\ZO$ (Section \ref{bqa}), (ii)
 Apply Lemma \ref{lemma:finalstates} to prove the implication
 $\con{zeroone} \Rightarrow \con{za}$ (Section \ref{zero}).
 (iii) Generalise the notion of zero automata, and prove $\con{za}
 \Leftrightarrow \con{quasi}$ (Section \ref{algorithm}).
In Section \ref{algorithm}, we will give a linear time algorithm (Theorem \ref{thm:algorithm}).
The logical aspects of our results are investigated in Section \ref{logics}.
Finally, we discuss some related works of our results and conclude this
paper in Section \ref{conclusion}.
We try to keep all sections as self-contained as possible.

%% file: chaps/preliminaries.tex
\section{Preliminaries}\label{preliminaries}
In this paper, all considered automata are {\it deterministic finite,
complete} and {\it accessible}. We refer the reader to the book by
Sakarovitch \cite{Sakarovitch:2009:EAT:1629683} for background material. \\

\sect{Languages and monoids}
We denote by $A^* \; [A^n]$ the set of all words [of length $n$] over a
nonempty finite alphabet $A$, and by $|w|$ the length of a word $w$ in $A^*$.
The empty word is denoted by $\varepsilon$.
That is, $A^*$ is the free monoid over $A$ with the neutral element $\varepsilon$.
We can easily verify that
\[
 \mu_{n+k}(A^k L) = \frac{|A^k L \cap A^{n+k}|}{|A^{n+k}|} =
 \frac{|A^k(L \cap A^{n})|}{|A^k A^n|} = \frac{|L \cap A^{n}|}{|A^{n}|}
 = \mu_n(L)
\]
holds for any language $L$ of $A^*$ and $k \geq 0$.
It follows from what has been said that $\mu(A^k L)$
exists if and only if $\mu(L)$ exists and in that case they are equal
$\mu(A^kL) = \mu(L)$.
If two languages $L$ and $K$ of $A^*$ are mutually
disjoint ($L \cap K = \emptyset$), then clearly $\mu(L \cup K) = \mu(L)
+ \mu(K)$ holds if both $\mu(L)$ and $\mu(K)$ exist.
We say that $v$ is a {\it factor of} $w$ if, there exists $x,y$ in $A^*$
such that $w = xvy$.
Let $L$ be a language of $A^*$ and let $u$ be a word of $A^*$.
The {\it left} [{\it right}] {\it quotient} $u^{-1} L \; [L u^{-1}]$ of $L$ by $u$ is
defined by
\[
 u^{-1} L = \{ v \in A^* \mid uv \in L \} \;\;\;\;\;\; \text{and}
 \;\;\;\;\;\; L u^{-1} = \{ v \in A^* \mid vu \in L\}.
\]
We denote by $\overline{L} = A^* \setminus L$ the {\it complement} of $L$.
The {\it syntactic congruence} of $L$ of $A^*$ is the relation $\sim_L$
defined on $A^*$ by $u \sim_L v$ if and only if, $xuy \in L
\Leftrightarrow xvy \in L$ holds for all $x, y$ in $A^*$.
The quotient $A^*/\sim_L$ is called the {\it syntactic monoid} of $L$ and the
natural morphism $\phi_L: A^* \rightarrow A^*/\sim_L$ is called the {\it
syntactic morphism} of $L$.
If $M$ is a monoid, an element $\bm{0}$ in $M$ is said to be a {\it zero} if,
$\bm{0}m = m\bm{0} = \bm{0}$ holds for all $m$ in $M$.\\

\sect{Automata and an important lemma}
An {\it (complete deterministic finite) automaton} over a finite alphabet $A$ is
a quintuple $\CA = \langle Q, A, \cdot, q_0, F \rangle$ where
\begin{itemize}
 \item $Q$ is a finite set of {\it states};
 \item $\cdot: Q \times A \rightarrow Q$ is a {\it transition function},
 which can be extended to a mapping $\cdot : Q \times A^* \rightarrow Q$
	   by $q \cdot \varepsilon = q$ and $q \cdot aw = (q \cdot a) \cdot
	   w$ where $q \in Q, a \in A$ and $w \in A^*$;
 \item $q_0 \in Q$ is an {\it initial state}, and $F \subseteq Q$ is a
	   set of {\it final states}.
\end{itemize}
The {\it language recognised by} $\CA$ is denoted by $L(\CA) = \{ w \in
A^* \mid q_0 \cdot w \in F\}$.
We say that $\CA$ {\it recognises} $L$ if $L = L(\CA)$.
It is a basic fact that, for any regular language $L$, there
exists a unique automaton recognises $L$ which has the minimum number of
states: the {\it minimal automaton} of $L$ and we denote it by $\CA_L$.
Each word $w$ in $A^*$ defines the transformation $w: q \mapsto q \cdot w$ on $Q$.
The {\it transition monoid} of $\CA$ is equal to the transformation
monoid generated by the generators $A$.
It is well known that the syntactic monoid of a regular language
is equal to the transition monoid of its minimal automaton.

For any subset $P$ of $Q$, the {\it past} of $P$
is the language denoted by $\past(P)$ and defined by
\[
 \past(P) = \{ w \in A^* \mid q_0 \cdot w \in P \}.
\]
Dually, the {\it future} of a subset $P$ of $Q$ is the language denoted
by $\future(P)$ and defined by
\[
 \future(P) = \{ w \in A^* \mid \exists p \in P, p \cdot w \in F \}.
\]
It is well known that, an (accessible) automaton $\CA$ is minimal if and
only if the following condition
\begin{equation}
 p = q \;\;\;\; \Leftrightarrow \;\;\;\; \future(p) = \future(q) \tag{{\bf M}} \label{condition:minimal}
\end{equation}
holds for every pair of states $p, q$ in $Q$.
Myhill-Nerode theorem states that every regular language has only a
finite number of left and right quotients.

In Section \ref{zero}, to prove Theorem \ref{thm:zero}, we will use the
following technical but important lemma. 
For the sake of completeness we include the proof, which is essentially
based on ``Proof of Theorem 3.2 and 3.2s'' in the book \cite{Eilenberg} by Eilenberg.
\begin{lemma}\label{lemma:finalstates}
 Let $\CA_L = \langle Q, A, \cdot, q_0, F \rangle$ be the minimal
 automaton of a language $L$. Then for any subset $P$ of $Q$, its past $\past(P)$
 can be expressed as a finite Boolean combination of languages of the form
 $Lw^{-1}$.
\end{lemma}
\begin{proof}
 We only have to prove that, for any state $q$ in $Q$, its past
 $\past(q)$ can be expressed as a Boolean combination of languages of
 the form $Lw^{-1}$. 
Our goal is to prove the following equation with the usual conventions $\bigcap_{w \in \emptyset} L w^{-1} = A^*$
 and $\bigcup_{w \in \emptyset} L w^{-1} = \emptyset$:
\begin{eqnarray}
 \past(q) = \left( \bigcap_{w \in \future(q)}    L w^{-1} \right) \setminus
            \left( \bigcup_{w \notin \future(q)} L w^{-1} \right). \label{eq:ba}
\end{eqnarray}
 The finiteness of this Boolean combination follows from Myhill-Nerode
 theorem.

 We prove first that the left hand side is contained in the right hand
 side in Equation \eqref{eq:ba}. 
 Let $v$ be a word in $\past(q)$. If a word $w$ in $\future(q)$, then
 $vw$ in $L$ by the definition, and hence $v$ in $L w^{-1}$.
 If a word $w$ not in $\future(q)$, then $vw$ not in $L$ by the
 definition, and hence $v$ not in $L w^{-1}$.
 It follows that the left hand side is contained in the right hand side
 in Equation \eqref{eq:ba}.

 Then we prove that the right hand side is contained in the left hand
 side in Equation \eqref{eq:ba}.
 Let $v$ be a word in right hand side in Equation \eqref{eq:ba}.
 Let $p$ be the state satisfies $q_0 \cdot v = p$, that is, $v$ is a
 word in $\past(p)$.
 For any $w$ in $\future(q)$, by the form of Equation \eqref{eq:ba}, $v$ is
 in $L w^{-1}$ from which we get $vw$ in $L$ whence $p \cdot w$ in $F$. That
 is, $w$ also belongs to $\future(p)$.
 Conversely, for any $w$ not in $\future(q)$, $vw$ is not in $L$ and
 thus $v$ not in $L w^{-1}$. That is, $w$ does not belong to $\future(p)$.
 It follows that $p$ and $q$ have the same future $\future(p) = \future(q)$
 from which we get $p = q$ by Condition \eqref{condition:minimal} of the
 minimality of $\CA_L$.
 Hence we obtain $v$ in $\past(q)$ and thus the right hand side
 is contained in the left hand side in Equation \eqref{eq:ba}. 
\end{proof}

\begin{remark}\upshape
 A {\it variety of languages} is a class of regular languages closed
under Boolean operations, left and right
 quotients and inverses of morphisms.
 The algebraic counterpart of a variety is a {\it (pseudo)variety of finite
 monoids}: a class of finite monoids closed under taking submonoids,
 quotients and finite direct products (\cf \cite{MPRI}).
 Eilenberg's variety theorem \cite{Eilenberg} states that varieties of
 languages are in one-to-one correspondence  with varieties of
 finite monoids.
 Lemma \ref{lemma:finalstates} shows us an importance of the Boolean operations taken in
 tandem with quotients.
 While this lemma is known (\cf \cite{Esik:2003:ActaCybernetica}),
 which is an ``automaton version'' of a key lemma in Eilenberg's variety
 theorem, we have not found any literature that includes a complete proof.
\end{remark}

%% file: chaps/synchronisation.tex
\section{Zero automata}\label{synchronisation}
In this seciton, we introduce a {\it zero automaton}, which plays a major
role in our work.
In contrast to the class of monoids with zero, their
natural counterpart, the class of zero automata has not been given much
attention.
To the best of our knowledge, only few studies (\eg \cite{Rystsov1997273})
have investigated zero automata in the context of the theory of
synchronising word for {\it \v{C}ern\'y's conjecture}.

Let $\CA$ be an automaton $\langle Q, A, \cdot, q_0, F
\rangle$.
For each pair of states $p,q$ in $Q$, we say that $q$ {\it is reachable
from} $p$ if, there exists a word $w$ such that $p \cdot w = q$. $\CA$
is called {\it accessible} if every state $q$ in $Q$ is reachable from
the initial state $q_0$.
A subset $P$ of $Q$ is called {\it strongly connected component}, if
for each state $q$ in $P$, $q$ is reachable from every other state in $P$.
A state $q$ in $Q$ is said to be {\it sink}, if $q \cdot a = q$ holds for
every letter $a$ in $A$.
We say that a subset $P$ of $Q$ is {\it sink}, analogously, if
there is no transition from any state $p$ in $P$ to a state which
does not in $P$. That is, $Q \setminus P$ are not reachable from $P$.
Note that, every (complete) automaton has at least one strongly
connected sink component. The family of all strongly connected
sink components of $\CA$ is denoted by $\sink(\CA)$.
A strongly connected component $P$ is {\it trivial} if it consists
of some single state $P = \{p\}$. 
We shall identify a singleton $\{p\}$ with its unique element $p$.
A word $w$ is a {\it synchronising word of} $\CA$ if, there exists a
certain state $q$ in $Q$, $p \cdot w = q$ holds for every state $p$ in
$Q$. That is, $w$ is the {\it constant map} from $Q$ to $q$.
We call an automaton {\it synchronising} if it has a synchronising word.
Note that any synchronising automaton has at most one sink state.
As we will prove in Section \ref{zero}, the following class of automata
captures precisely the zero-one law for regular languages.
\begin{definition}[\cite{Rystsov1997273}]\label{def:sink}\upshape
A {\it zero automaton} is a synchronising automaton with a sink state.
\end{definition}

\begin{example}\label{ex:zeroautomata}\upshape
\begin{figure}[t]
\centering\includegraphics[width=.8\columnwidth]{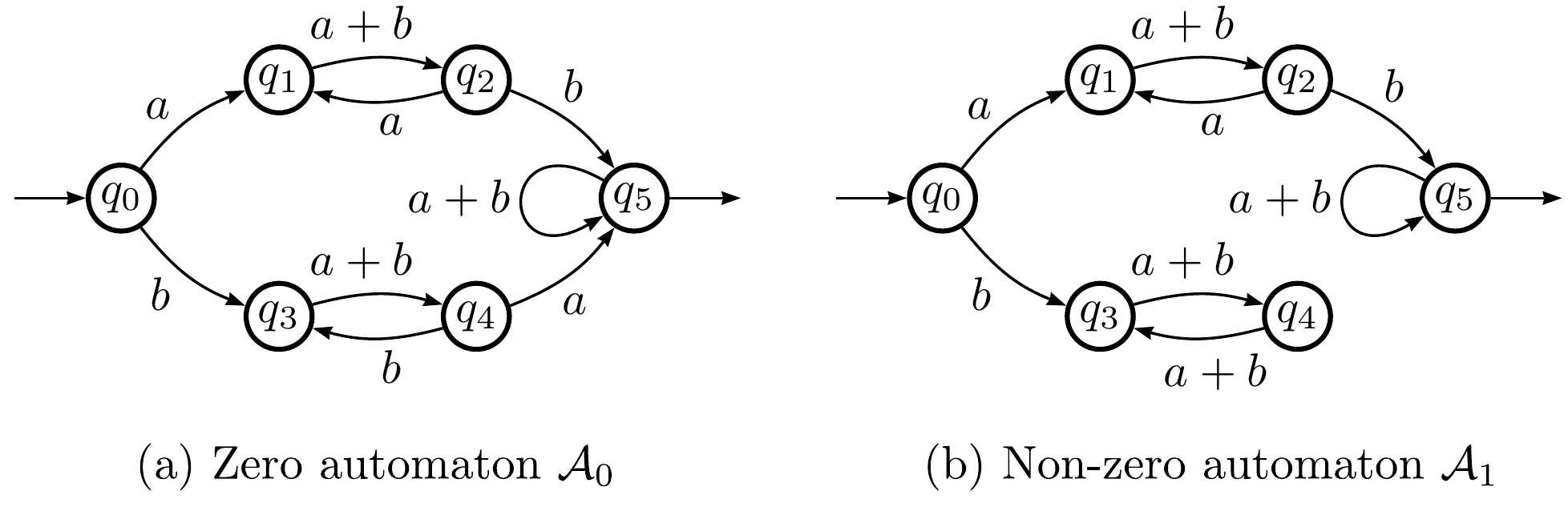}
\caption{Zero and non-zero automata}
\label{fig:zero-one}
\end{figure}
Consider two automata $\CA_0$ and $\CA_1$ illustrated in Figure
 \ref{fig:zero-one}.
$\CA_0$ is a zero automaton but $\CA_1$ is not, though both automata have a
 sink state $q_5$. The only difference between $\CA_0$ and $\CA_1$ is
 the transition result of $q_4 \cdot a$; which equals to $q_5$ in $\CA_0$,
 while which equals to $q_3$ in $\CA_1$.
We can easily verify that, $\CA_0$ has a unique strongly
 connected sink component $q_5$, while $\CA_1$ has two 
 strongly connected sink components $\{q_3, q_4\}$ and $q_5$.
\end{example}
Definition \ref{def:sink} can be rephrased as follows.
\begin{lemma}\label{lemma:synchronisation}
Let $\CA = \langle Q, A, \cdot, q_0, F \rangle$ be an automaton.
 Then $\CA$ is zero if and only if $\CA$ has a unique strongly connected
 sink component and it is trivial, \ie $\sink(\CA) = \{\{p\}\}$ for
 a certain sink state $p$.
\end{lemma}
\begin{proof}
 First we assume $\CA$ is zero with a sink state $p$. Then there exists a
 synchronising word $w$ and it clearly satisfies $q \cdot w = p$ for each
 $q$ in $Q$ since $p$ is sink. This shows that there is no strongly
 connected sink component in $Q \setminus p$.

 Now we prove the converse direction, we assume $\CA$ has a unique
 strongly connected sink component and it is trivial, say $p$.
 We can verify that for every state $q$ in $Q$, there exists a word $w$
 in $A^*$, such that $q \cdot w = p$.
 Indeed, if there does not exist such word $w$ for some $q$, then the
 set of all reachable states from $q: \{r \in Q \mid \exists w \in
 A^*, q \cdot w = r \}$ must contains at least one strongly connected
 sink component which does not contain $p$. This contradicts with the
 uniqueness of the closed strongly connected component $p$ in $\CA$.
 The existence of a synchronising word $w$ is guaranteed, because we
 can concretely construct it as follows.
 Let $n$ be the number of states $n = |Q|$ and let $Q = \{q_0, \cdots,
 q_{n-1} = p\}$. We define a word sequence $w_i$ inductively by
 $w_0 = u_{q_0}$ and $w_{i} = u_{(q_i \cdot v_{i-1})}$ where each $u_{q_i}$ is
 a shortest word satisfies $q_i \cdot u_{q_i} = p$, and $v_{i-1}$ is the
 word of the form $w_0 \cdots w_{i-1}$.
\begin{figure}[t]
\centering\includegraphics[width=0.6\columnwidth]{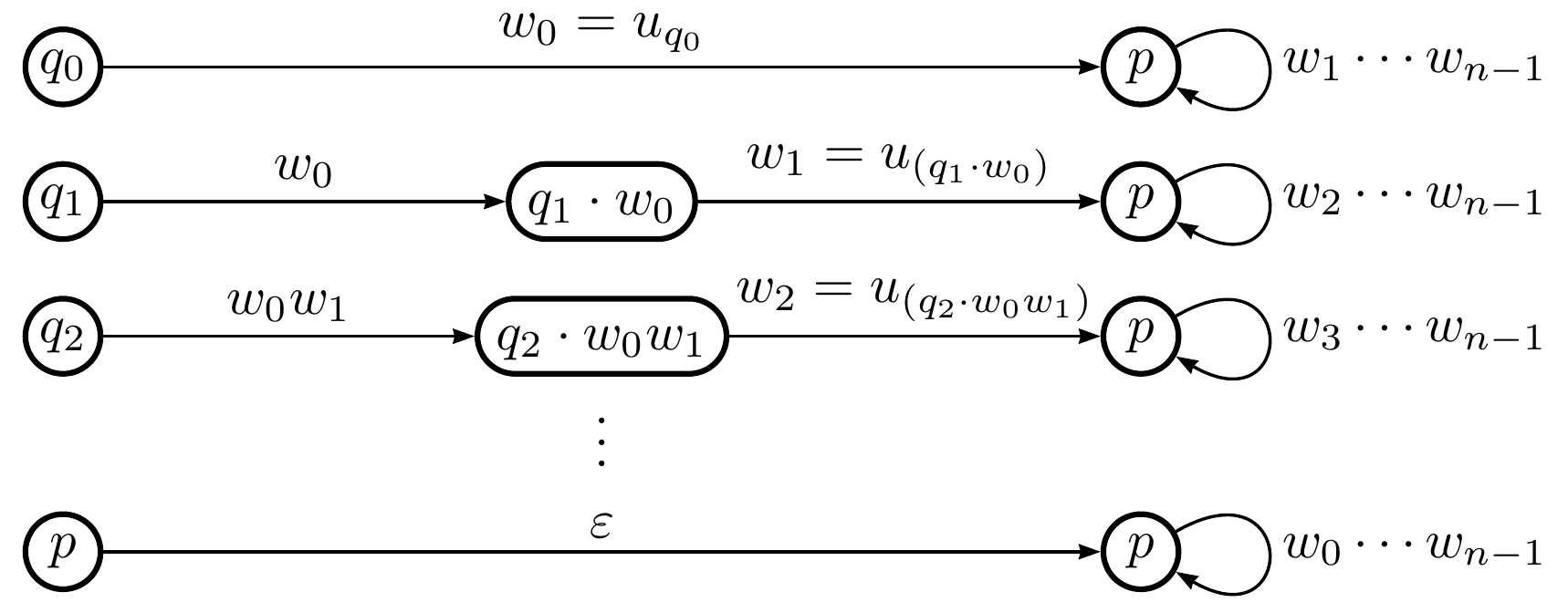}
\caption{Synchronising word $v_{n-1} = w_0 \cdots w_{n-1}$ in the
 proof of Lemma \ref{lemma:synchronisation}}
 \label{fig:synchronisation}
\end{figure}
 As shown in Figure \ref{fig:synchronisation}, we can easily verify that
 the word $v_{n-1} = w_0 \cdots w_{n-1}$ is a synchronising word
 satisfies $q \cdot v_{n-1} = p$ for each $q$ in $Q$. 

 For example, consider the zero automaton $\CA_0$ in Figure
 \ref{fig:zero-one}. Then each $u_{q_i}, w_{q_i}$ and $v_{q_i}$ are defined as follows.
\begin{center}
\begin{tabular}{cccc}
 & \;\;\; $u_{q_i}$ \vspace{.5mm} \;\;\; & \;\;\; $w_{q_i}$ \;\;\; & \;\;\; $v_{q_i}$ \;\;\; \\\hline
\;\;\; $q_0$ & $aab$ & $aab$ & $aab$ \\
\;\;\; $q_1$ & $ab$ & $b$ & $aabb$ \\
\;\;\; $q_2$ & $b$ & $\varepsilon$ & $aabb$ \\
\;\;\; $q_3$ & $aa$ & $\varepsilon$ & $aabb$ \\
\;\;\; $q_4$ & $a$ & $\varepsilon$ & $aabb$ \\
\;\;\; $q_5$ & $\varepsilon$ & $\varepsilon$ & $aabb$ \\
\end{tabular}
\end{center}
 The obtained word $v_{q_4} = aabb$ is a synchronising word which
 satisfies $q_i \cdot aabb = q_5$ for all $q_i$ in $\CA_0$.
It is clear that the non-zero automaton $\CA_1$ in Figure \ref{fig:zero-one} does not have
 a synchronising word since it has two strongly connected sink components.
\end{proof}

%% file: chaps/bqa.tex
\section{Closure properties of $\ZO$}\label{bqa}
We first introduce the following lemma.
\begin{lemma}\label{lemma:probability}
Let $L$ be a language of $A^*$ and $w$ be a word in $A^k$.
 Then the asymptotic probability of $L$ exists if and only if the asymptotic probability of the language $wL \;
 [Lw]$ exists. Moreover, these limits satisfies the equation $\mu(wL) = \mu(Lw) = |A|^{-k}\mu(L)$.
\end{lemma}
\begin{proof}
 Since $wL$ and $Lw$ clearly have the same counting function, we only have to
 prove the case of $wL$. 
For every $u, v$ in $A^k$ such that $u \neq v$, the
 language $uL$ and $vL$ are obviously mutually disjoint and these counting
 functions satisfies
\[
 \gamma_n(uL) = \gamma_n(vL) = \begin{cases}
								0 &  n < k,\\
								\gamma_{n-k}(L) & n \geq k.
							   \end{cases}
\]
This shows that $uL$ and $vL$ have the same counting function and thus
 have the same asymptotic probability if its exists.
We can easily verify that 
\[
 \mu(L) = \mu\left( A^k L \right) = \sum_{u \in A^k}
 \mu(uL) = |A|^k \mu(wL)
\]
holds for any $w$ in $A^k$.
\end{proof}

Now we prove the following proposition, which states the necessary closure
properties of the class $\ZO$ for Lemma \ref{lemma:finalstates}.
\begin{proposition}\label{prop:bqa}
 $\ZO$ is closed under Boolean operations, left and right quotients.
\end{proposition}

\begin{proof}[Proposition \ref{prop:bqa}] We first prove that $\ZO$ is closed under Boolean
 operations, and then prove that $\ZO$ is closed under quotients.\\

\sect{$\ZO$ is closed under Boolean operations}
Let $L, K$ be two languages in $\ZO$.
It is obvious that $\ZO$ is closed under complement since
 $\mu(\overline{L}) = 1 - \mu(L) \in \{0, 1\}$, and we can easily verify
 that the following equations holds.
\begin{itemize}
 \item $\mu(L \cup K) = 0$ if $\mu(L) = 0$ and $\mu(K) = 0$;
 \item $\mu(L \cap K) = 0$ if either $\mu(L) = 0$ or $\mu(K) = 0$;
 \item $\mu(L \cup K) = 1$ if either $\mu(L) = 1$ or $\mu(K) = 1$;
 \item $\mu(L \cap K) = 1$ if $\mu(L) = 1$ and $\mu(K) = 1$.
\end{itemize}\vspace{5mm}

\sect{$\ZO$ is closed under quotients}
We first prove that $\ZO$ is closed under left quotients.
Let $L$ be a regular language in $\ZO$ and we assume that $L$ does
 not contain $\varepsilon$ without loss of generality.
First we assume $\mu(L) = 0$.
By the definition of left quotients, one can easily verify that
\[
 L = \bigcup_{a \in A} L \cap aA^* = \bigcup_{a \in A} a a^{-1} L
\]
holds (since $\varepsilon \notin L$) and all these sets $aa^{-1}L$ $(= L
 \cap aA^*)$ are mutually disjoint. It follows that the following equation holds.
\begin{eqnarray*}
 \mu(L) &=& \limn \frac{|L \cap A^n|}{|A^n|} = \limn
  \frac{|\left(\bigcup_{a \in A} a a^{-1} L\right) \cap A^n|}{|A^n|} 
 = \limn \frac{|\bigcup_{a \in A} (a a^{-1} L \cap A^n)|}{|A^n|}
 \nonumber\\
 &=& \limn \sum_{a \in A}\frac{|a a^{-1} L \cap A^n|}{|A^n|} = \sum_{a \in A} \mu(a a^{-1} L) = 0.
\end{eqnarray*}
That is, the asymptotic probability $\mu(a a^{-1} L)$
 equals to zero for each $a$ in $A$, since these summation converges to zero.
In addition, $\mu(aa^{-1}L)$ coincides with $\mu(a^{-1} L)$ for any $a$
 in $A$, because $\mu(aa^{-1}L)  = |A|^{-1}\mu(a^{-1}L) = 0$ by Lemma
 \ref{lemma:probability} whence $\mu(a^{-1} L) = 0$.

Next we assume $\mu(L) = 1$. Then $\mu(\overline{L}) = 0$ and 
\begin{eqnarray*}
 a^{-1} \overline{L} = \{ w \in A^* \mid aw \in \overline{L} \} = \{ w
  \in A^* \mid aw \notin L \} = \overline{a^{-1} L}
\end{eqnarray*}
holds. We therefore obtain:
\[
 \mu(a^{-1}L) = 1 - \mu(\overline{a^{-1}L}) = 1 -
 \mu(a^{-1}\overline{L}) = 1 - 0 = 1.
\]
We can prove that $\ZO$ is closed under right quotients by the same
 manner. 
\end{proof}

%% file: chaps/zero.tex
\section{Equivalence of $\ZO$ and $\Z$}\label{zero}
We will use the following lemma, which is a direct consequence of Lemma
\ref{lemma:finalstates} and Proposition \ref{prop:bqa}.
\begin{lemma}\label{lemma:everywarezeroone}
 Let $L$ be a regular language in $\ZO$, $\CA_L = \langle Q, A, \cdot,
 q_0, F \rangle$ be its minimal automaton.
 Then, for any subset $P$ of $Q$ in $\CA_L$, its past $\past(P)$ is also
 in $\ZO$.
\end{lemma}
\begin{proof}
 By Lemma \ref{lemma:finalstates}, for any subset $P$ of $Q$, its past
 $\past(P)$ can be expressed as a finite Boolean combination of languages of
 the form $Lw^{-1}$.
 It follows that $\past(P)$ obeys the zero-one law, since $L$ is in $\ZO$ and $\ZO$
 is closed under Boolean operations and quotients by Proposition \ref{prop:bqa}.
\end{proof}
Lemma \ref{lemma:everywarezeroone} will be used for proving the
direction $\con{zeroone} \Rightarrow \con{za}$. Now we give a proof.

\begin{proof}[Proof of Theorem \ref{thm:zero}]
We show the implication $\con{za} \Rightarrow
 \con{zero} \Rightarrow \con{zeroone} \Rightarrow
 \con{za}$.
The former implication $\con{za} \Rightarrow \con{zero}
 \Rightarrow \con{zeroone}$ is easy and almost folklore, but we include
 a proof here to be self-contained.\\

\sect{$\con{za} \Rightarrow \con{zero}$ ($\CA_L$ is zero $\Rightarrow L$
 is with zero)}
Let $\CA_L = \langle Q, A, \cdot, q_0, F \rangle$ be the minimal
 automaton of $L$ and it is zero with a sink state $p$.
Let $M$ be the transition monoid of $\CA_L$ and $\phi: A^* \rightarrow
 M$ be the syntactic morphism of $L$.
 Then we can verify that $M$ has a zero element $\bm{0}$ as the
 transformation $\bm{0}: q \mapsto p$ for all $q$ in $Q$,
 that is, $\bm{0}$ is the constant map from $Q$ to $p$.
 The existence of $\bm{0}$ is guaranteed since $\CA_L$ is synchronising.
 Indeed, for any synchronising word $w$, $\phi(w) = \bm{0}$ holds.
 One can easily verify that $m\bm{0} = \bm{0}m = \bm{0}$ for all $m$ in $M$.
 This proves that $M$ the syntactic monoid of $L$ has the zero.\\

\sect{$\con{zero} \Rightarrow
 \con{zeroone}$ ($L$ is with zero $\Rightarrow L$ obeys the zero-one law)}
Let $L$ be a regular language in $\Z$, $M$ be its syntactic monoid with
 a zero element $\bm{0}$ and $\phi: A^* \rightarrow M$ be its syntactic
 morphism.
We choose a word $w_{\bm{0}}$ from the preimage of $\bm{0}$: $w_{\bm{0}} \in
 \phi^{-1}(\bm{0})$.

Now we prove $\mu(L) = 1$ if $w_{\bm{0}}$ in $L$.
By the definition of zero, we have
\[
 \phi(x w_{\bm{0}} y) = \phi(x)\phi(w_{\bm{0}})\phi(y) = \phi(x)\bm{0}\phi(y) = \bm{0} 
\]
 for any words $x,y$ in $A^*$. That is, if $w$ contains $w_{\bm 0}$ as
 a factor, then $\phi(w) = \phi(w_{\bm{0}}) = \bm{0}$ holds and hence
 $w$ also in $L$.
Let $L_{w_{\bm{0}}} = A^* w_{\bm{0}} A^*$ be the set of all words that contain $w_{\bm{0}}$
 as a factor.
Then clearly $L_{w_{\bm{0}}}$ is contained in $L$ from
 which we get $\mu_n(L_{w_{\bm{0}}}) \leq \mu_n(L)$ for
 all $n$.
The probability $\mu_n(L_{w_{\bm{0}}})$ is nothing but the
 probability that a randomly chosen word of length $n$ contains
 $w_{\bm{0}}$ as a factor.
 The following well known elementally fact, sometimes called {\it
 Borges's theorem} (\cf Note I.35 in \cite{Flajolet:2009:AC:1506267}),
 ensures that $\mu_n(L_{w_{\bm{0}}})$ tends to one if $n$ tends to
 infinity. This shows $\mu(L) = \mu(L_{w_{\bm{0}}}) = 1$ and we can
 prove $\mu(L) = 0$ if $w_{\bm 0}$ not in $L$ by the same manner.

 \vspace{2mm}\noindent{\bf Borges's theorem.}
 {\it   Take any fixed finite set $\Pi$ of words in $A^*$. A random word
 in $A^*$ of length $n$ contains all the words of the set $\Pi$ as
 factors with probability tending to one exponentially fast as $n$
 tends to infinity. }\vspace{5mm}

\sect{$\con{zeroone} \Rightarrow
 \con{za}$ ($L$ obeys the zero-one law $\Rightarrow \CA_L$ is zero)}
Let $L$ be a regular language in $\ZO$ and $\CA_L = \langle Q, A, \cdot,
 q_0, F \rangle$ be its minimal automaton, let $\sink(\CA_L) = \{P_1, \cdots, P_k
 \}$ for some $k \geq 0$.
 Our goal is to prove $k = 1$ and $\sink(\CA_L) = \{\{p\}\}$ for a
 certain sink state $p$. It follows that $\CA_L$ is zero by Lemma
 \ref{lemma:synchronisation}.

 For any strongly connected sink component $P_i$,
 there exists a word $w_i$ such that $q_0 \cdot w_i$ in $P_i$
 because $\CA_L$ is accessible. Since $P_i$ is sink, the language
 $w_i A^*$ is contained in $\past(P_i)$ from which we get
\begin{eqnarray}
 0 < \mu(w_i A^*) = |A|^{-|w_i|} \mu(A^*) = |A|^{-|w_i|} \leq
  \mu(\past(P_i)) \label{eq:past_}
\end{eqnarray}
 for each $P_i$ by Lemma \ref{lemma:probability}.
 Lemma \ref{lemma:everywarezeroone} and Equation \eqref{eq:past_} implies that
 the asymptotic probability $\mu(\past(P_i))$ surely exists and satisfies
 \begin{eqnarray}
  \mu(\past(P_i)) = 1 \label{eq:past}
 \end{eqnarray}
 for every strongly connected sink component $P_i$.

 Now we prove $k = 1$.
 By Equation \eqref{eq:past}, we can easily verify that
 \begin{eqnarray*}
  \mu\left( \bigcup_{i = 1}^k \past(P_i) \right) = \sum_{i = 1}^k
  \mu(\past(P_i)) = k
 \end{eqnarray*}
 holds because $\CA_L$ is deterministic and thus all $\past(P_i)$ are
 mutually disjoint.
 This clearly shows $k = 1$, that is, there exists a unique strongly
 connected sink component, say $P$, in $\CA_L$: $\sink(\CA_L) = \{P\}.$

 Next we let $P = \{p_1, \cdots, p_n\}$ and prove $n = 1$. 
 Since $P$ satisfies $\mu(\past(P)) = 1$ by Equation \eqref{eq:past},
 there exists exactly one state $p$ in $P$ satisfies $\mu(\past(p)) =
 1$ by Lemma \ref{lemma:everywarezeroone}.
 Further, because $P$ is strongly connected, for every state $p_i$ in $P$, there
 exists a word $w_i$ such that $p \cdot w_i = p_i$.
 It follows that $\past(p) w_i \subseteq \past(p_i)$ and thus
\begin{eqnarray}
 0 < \mu(\past(p) w_i) = |A|^{-|w_i|} \mu(\past(p)) = |A|^{-|w_i|} \leq
  \mu(\past(p_i)) = 1 \label{eq:pastp}
\end{eqnarray}
 holds for every state $p_i$ in $P$ by Lemma \ref{lemma:probability} and
 Lemma \ref{lemma:everywarezeroone}.
 Equation \eqref{eq:past} and \eqref{eq:pastp} implies
 \begin{eqnarray*}
 \mu(\past(P)) = \sum_{i = 1}^n \mu(\past(p_i)) = \sum_{i = 1}^n 1 = n = 1,
 \end{eqnarray*}
 because $\CA_L$ is deterministic and thus all $\past(p_i)$ are
 mutually disjoint.
 We now obtain $n = 1$, that is, $P$ is singleton and hence
 $\sink(\CA_L) = \{p\}$. That is, $\CA_L$ is zero.
\end{proof}

\begin{remark}\upshape
 It is interesting that, though we use Borges's theorem to prove the
 direction $\con{zero} \Rightarrow \con{zeroone}$,
 Theorem \ref{thm:zero} is a vast generalisation of Borges's theorem,
 since any language of the form $A^* K A^*$ where $K$ is regular is
 always recognised by a zero automaton (but the converse is not true). 
 To state Theorem \ref{thm:zero} more precisely, by the proof above we
 can easily verify that, a zero-one language $L$ satisfies $\mu(L) = 1
 \; [\mu(L) = 0]$ if and only if its minimal automaton $\CA_L$ is zero
 and the sink state of $\CA_L$ is final [non-final].
\end{remark}

%% file: chaps/algorithm.tex
\section{Linear time algorithm for testing the zero-one law}\label{algorithm}
The equivalence of zero-automata and the zero-one law gives us an
effective algorithm.
For a given $n$-states automaton $\CA$, we can determine whether
$L(\CA)$ obeys the zero-one law by the following steps: (i) Minimise $\CA$
to obtain its minimal automaton $\CB$. (ii) Calculate the family of all
strongly connected components $P$ of $\CB$. (iii) Check whether $P$
contains exactly one strongly connected sink component and it is
trivial, \ie whether $\CB$ is a zero automaton (Lemma
\ref{lemma:synchronisation}).
It is well known that Hopcroft's automaton minimisation algorithm has 
an $\order(n \log n)$ time complexity and Tarjan's strongly connected components
algorithm has an $\order(n + n|A|) = \order(n)$ complexity where $n |A|$ means the
number of {\it edges}. Hence we can minimise $\CA$ to obtain $\CB$ in
 $\order(n \log n)$ on the step (i), and can calculate $P$ in $\order(n)$
 on the step (ii). One can easily verify that the step (iii) above can be
 done in $\order(n)$. To sum up, we have an $\order(n \log n)$ algorithm
 for testing whether a given  regular language obeys the
 zero-one law, if its is given by an $n$-states deterministic finite automaton. 
We can obtain, however, more efficient algorithm {\it by avoiding minimisation}.
In order to do that, there is a need for further investigation of
the structure of zero automata.\\

\sect{Quasi-zero automata and more effective algorithm}
Let $\CA = \langle Q, A, \cdot, q_0, F \rangle$ be an automaton.
The {\it Nerode equivalence} $\sim$ of $\CA$ is the relation defined on
$Q$ by $p \sim q$ if and only if $\future(p) = \future(q)$.
One can easily verify that $\sim$ is actually a congruence, in the sense
that $F$ is saturated by $\sim$ and $p \sim q$ implies $p \cdot w \sim
q \cdot w$ for all $w \in A^*$. Hence it follows that there is a well
defined new automaton $\CA\Q$, {\it the quotient automaton of}
$\CA$:
\[
 \CA\Q = \langle Q\Q, A, \cdot, [q_0]_{\sim}, F\Q \rangle
\]
where $[q]_{\sim}$ is the equivalence class modulo $\sim$ of $q$,
$S\Q = \{[q]_{\sim} \mid q \in S\}$ is the set of the equivalence
classes modulo $\sim$ of a subset $S \subseteq Q$, and where the
transition function $\cdot: Q\Q \times A \rightarrow Q\Q$
is defined by $[p]_\sim \cdot a = [p\cdot a]_\sim$.
We define the natural mapping $\phi_\sim: Q \rightarrow Q\Q$ by
$\phi_\sim(q) = [q]_\sim$.
Condition \eqref{condition:minimal} for minimal automata
implies that, for any automaton $\CA$, its quotient automaton $\CA\Q$
is the minimal automaton of $L(\CA)$. We shall identify
the quotient automaton $\CA\Q$ with the minimal automaton of $L(\CA)$ (\cf
\cite{Sakarovitch:2009:EAT:1629683}).

We now introduce a new class of automata which is a generalisation of
the class of zero automata.
\begin{definition}[quasi-zero automaton]\upshape
An automaton $\CA = \langle Q, A, \cdot, q_0, F \rangle$ is {\it
 quasi-zero} if either $\bigcup\sink(\CA)$ $\subseteq$ $F$ or
 $\bigcup\sink(\CA) \cap F = \emptyset$ holds.
\end{definition}
Since every zero automaton $\CA$ satisfies $\bigcup\sink(\CA) = \{p\}$ for a
certain state $p$ (Lemma \ref{lemma:synchronisation}), every zero
automaton is quasi-zero.
The following proposition shows that the minimal automaton of any
quasi-zero automaton is zero and {\it vice versa} (this justifies the
term ``quasi-zero'').

\begin{proposition}\label{prop:quasizero}
An automaton $\CA = \langle Q, A, \cdot, q_0, F \rangle$ is quasi-zero
 if and only if $\CA\Q$ is zero.
\end{proposition}
\begin{proof}
This proposition shows exactly the equivalence $\con{za} \Leftrightarrow
 \con{quasi}$ in Theorem \ref{thm:zero}.\\

\sect{$\con{za} \Rightarrow \con{quasi}$ ($\CA\Q$ is zero $\Rightarrow \CA$ is quasi-zero)}
Let $p$ be the unique sink state of $\CA\Q$.
To prove this direction, it is enough to consider the case when
 $p \in F\Q$, \ie $\future(p) = A^*$.
We now show
\begin{eqnarray}
  \bigcup\sink(\CA) \subseteq F \label{eq:qac}
\end{eqnarray}
 by contradiction.
Let us assume that Inclusion \eqref{eq:qac} does not hold, that is, we
 assume there exists a non-final state $q$ in $\bigcup\sink(\CA)$.
 Let $P$ be the strongly connected sink component of $\CA$ that
 contains $q$. Since $P$ is sink and strongly connected, $\phi_\sim(P)$ is
 sink and strongly connected in $\CA\Q$ too.
 Moreover, $\phi_\sim(P)$ does not contain the sink state $p$, because
 $q \notin F$ implies that, for any state $q'$ in $P$, $\future(q') \neq A^*$
 from which we obtain $\future([q']_\sim) \neq \future(p)$ and $[q']_\sim \neq
 p$. That is, $\CA\Q$ has at least two strongly connected sink
 components $\phi_\sim(P)$ and $p$. This is contradiction.\\

\sect{$\con{quasi} \Rightarrow \con{za}$ ($\CA$ is quasi-zero $\Rightarrow \CA\Q$ is zero)}
To prove this direction, it is enough to consider the case when
 $\bigcup\sink(\CA) \subseteq F$.
Since $\CA$ is quasi-zero, all states in $\bigcup\sink(\CA)$ have the same
 future $A^*$, \ie $\future(q) = A^*$ for every state $q$ in
 $\bigcup\sink(\CA)$, because $\bigcup\sink(\CA) \subseteq F$ implies $q
 \cdot w \in F$ for every state $q$ in $\bigcup\sink(\CA)$ and every word $w$
 in $A^*$.
 This implies that $\bigcup\sink(\CA)\Q$ consists of a
 single equivalence class, say $p$.
 Moreover, this equivalence class $p$ is a sink state in $\CA\Q$ by the
 definition of sink and Condition \eqref{condition:minimal} of the
 minimality of $\CA\Q$.
We now show that, by contradiction, $\CA\Q$ has only one strongly
 connected sink component $p$:
\begin{eqnarray}
  \bigcup\sink(\CA\Q) = \{p\} \label{eq:caq}
\end{eqnarray}
 from which we obtain $\CA\Q$ is zero by Lemma \ref{lemma:synchronisation}.
Let us assume that Inclusion \eqref{eq:caq} does not hold, that is, we
 assume there exists a strongly connected sink component $R = \{r_1,
 \cdots, r_n\}$ of $\CA\Q$, which does not contain $p$. Recall that each state $r_i$ of
 $\CA\Q$ is an equivalence class, \ie a set of states, of $\CA$.
 Let $S = \phi_\sim^{-1}(R)$ be a set of states of $\CA$.
 Since $R$ is strongly connected sink component of $\CA\Q$, its preimage
 $S$ contains at least one strongly connected sink component, say $P$,
 of $\CA$.
 For every state $q$ in $P$, $\future(q)$ is not equal to $A^* = \future(p)$,
 because $p \notin \phi_\sim(P) \subseteq \phi_\sim(S) = R$ implies $[q]_\sim \neq p$.
 This contradicts with the assumption that $\future(q) = A^*$ for every state $q$ in
 $\bigcup\sink(\CA)$.
 This completes the proof of Theorem \ref{thm:zero}.
\end{proof}

By using this proposition, we obtain a linear time algorithm by
avoiding minimisation as stated in the following theorem.
\begin{theorem}\label{thm:algorithm}
There is an $\order(n)$ algorithm for testing whether a given regular
 language is zero-one, if its is given by an $n$-states deterministic
 finite automaton.
\end{theorem}
\begin{proof}
For a given $n$-states automaton $\CA$, we can determine whether
$L(\CA)$ obeys the zero-one law by the following steps: (i) Calculate the
 family of all strongly connected components $P$ of $\CA$.
(ii) Extract all strongly connected sink components from $P$ to obtain
 $\sink(\CA)$.
(iii)  Check whether, in $\bigcup\sink(\CA)$, either all states are
 final or all states are non-final, \ie whether $\CA$ is quasi-zero.
 By Theorem \ref{thm:zero}, $L(\CA)$ obeys the zero-one law if and only
 if $\CA$ is quasi-zero. Hence this algorithm is correct. All steps (i)
 $\sim$ (iii) can be done in $\order(n)$, this ends the proof.
\end{proof}

%% file: chaps/logics.tex
\section{Logical aspects of the zero-one law}\label{logics}
There are different manners to define a language: a set of {\it finite
words}.
In the descriptive approach, the words of a language are characterised
by a property. The automata approach is a special case of the
descriptive approach.
Another variant of the descriptive approach consists in defining
languages by logical formulae: we regard words as {\it finite structures
with a linear order composed of a sequence of positions labeled over
finite alphabet}.
The zero-one law, which is defined in this paper, has been studied
extensively in finite model theory (\cf Chapter 12 ``Zero-One Laws'' of
\cite{Libkin:2004:EFM:1024196}).
This notion can be applied to logics over, not only finite words,
but also arbitrary {\it finite structures}, such as {\it finite graphs}:
we regard graphs as finite structures with a set of nodes and their edge
relation.
We say that a logic ${\cal L}$, over fixed finite structures, has the
zero-one law if every property $\Phi$ definable in ${\cal L}$ satisfies
$\mu(\Phi) \in \{0, 1\}$ ($\mu$ is defined analogously).
Broadly speaking, every property $\Phi$ is either {\it almost surely
true} or {\it almost surely false}.
Fagin's theorem \cite{DBLP:journals/jsyml/Fagin76} states that {\it
first-order logic} $\FO$ for finite graphs has the zero-one
law. Moreover, an $\FO$ sentence $\Phi$ is almost surely true (\ie
$\mu(\Phi) = 1$) if and only if $\Phi$ is true on a certain infinite
graph: the {\it random graph}.
This characterisation leads to the fact that, for any $\FO$ sentence $\Phi$,
it is decidable whether $\mu(\Phi) = 1$ (\cf Corollary 12.11 in
\cite{Libkin:2004:EFM:1024196}).
After the work of Fagin, much ink has been spent on the zero-one law for
logics over finite graphs. It is now known that many logics (\eg {\it
logic with a fixed point operator} \cite{DBLP:journals/iandc/BlassGK85},
{\it finite variable infinitary logic} \cite{Kolaitis1992258} and
certain fragments of {\it second-order logic}
\cite{conf/mfcs/KolaitisV00}) have the zero-one law.

By contrast, though many logics have the zero-one law, their extensions
with ordering (like as logics over finite words), no longer have it.
In fact, over both finite graphs and finite words, while first-order
logic $\FO$ has the zero-one law, its extension with a linear order
$\FO[<]$ does not.

\begin{example}\label{ex:logic}\upshape
A simple counterexample is the language $(aA)^*$ which can be defined by
the $\FO[<]$ sentence $\Phi_{aA^*} = \exists i \left( \forall
j (i < j) \land P_a(i)  \right).$
The variables $i$ and $j$ of this sentence represent {\it position} in a
word. The sentence $P_a(i)$ is interpreted to mean ``the $i$-th letter
is $a$''. This language $aA^*$ satisfies $\mu_n(aA^*) = 1 / |A|$ as
we stated in Section \ref{introduction}, hence $\Phi_{aA^*}$ does not
obey the zero-one law in general.
It follows that $\FO[<]$ for finite words does not have the zero-one
law. 
\end{example}

We summarise well known logical and algebraic characterisations of
classes of languages, including the class of zero-one languages $\ZO$, in
Figure \ref{fig:separation}.
Details and full proofs of these results can be found in a very nice
survey \cite{DGK-ijfcs08} by Diekert {\it et al}. 
In Figure \ref{fig:separation}, we use standard abridged notation:
$\FO^n[<]$ for first-order logic with $n$ variables; $\Sigma_n[<]$ for
$\FO$ formulae with $n$ blocks of quantifiers and starting with a block
of existential quantifiers; $\bool\Sigma_n[<]$ for the Boolean closure
of $\Sigma_n[<]$.
A {\it monomial} over $A$ is a language of the form $A_0^* a_1 A_1^* a_2
\cdots a_k A_k^*$ where $a_i$ in $A$ and $A_i \subseteq A$ for each $i$,
and is {\it unambiguous} if for all $w \in A_0^* a_1 A_1^* a_2
\cdots a_k A_k^*$ there exists exactly one factorisation $w = w_0 a_1
w_1 a_w \cdots a_k w_k$ with $w_i$ in $A_i^*$ for each $i$.
A language $L$ over $A$ is called:
\begin{itemize}
\item  {\it star-free} if it is expressible by union, concatenation and
	   complement, but does not use Kleene star;
 \item  {\it polynomial} if it is a finite union of monomials;
 \item  {\it unambiguous polynomial} if it is a finite disjoint union of
  unambiguous monomials;
 \item  {\it piecewise testable} if it is a finite Boolean combination
		of simple polynomials;
 \item  {\it simple polynomial} if it
is a finite union of languages of the form $A^* a_1 A^* a_2 \cdots a_k
A^*$.
\end{itemize}

\begin{figure}[t]
\begin{minipage}[t]{0.55\linewidth}
\centering{\renewcommand\arraystretch{1.5}
\scriptsize{\begin{tabular}{ccc}
\Hline
Languages & Monoids & Logic \\\Hline
regular & finite & $\MSO[<]$\\
star-free & {\it aperiodic}  & $\FO[<]$\\
polynomials &  & $\Sigma_2[<]$\\
unambiguous polynomials & $\bold{DA}$ & $\FO^2[<]$\\\hline
zero-one  & with zero & ? \\\hline
piecewise testable & ${\cal J}$-trivial & $\bool \Sigma_1[<]$ \\
simple polynomial & & $\Sigma_1[<]$ \\
$\bool \{A^* \mid A \subseteq \Sigma \}$ & commutative and idempotent & $\FO^1[<]$ \\
\Hline
\end{tabular}}}
\end{minipage}
\begin{minipage}{0.45\linewidth}
\vspace{7mm}
 \centering\includegraphics[width=1\columnwidth]{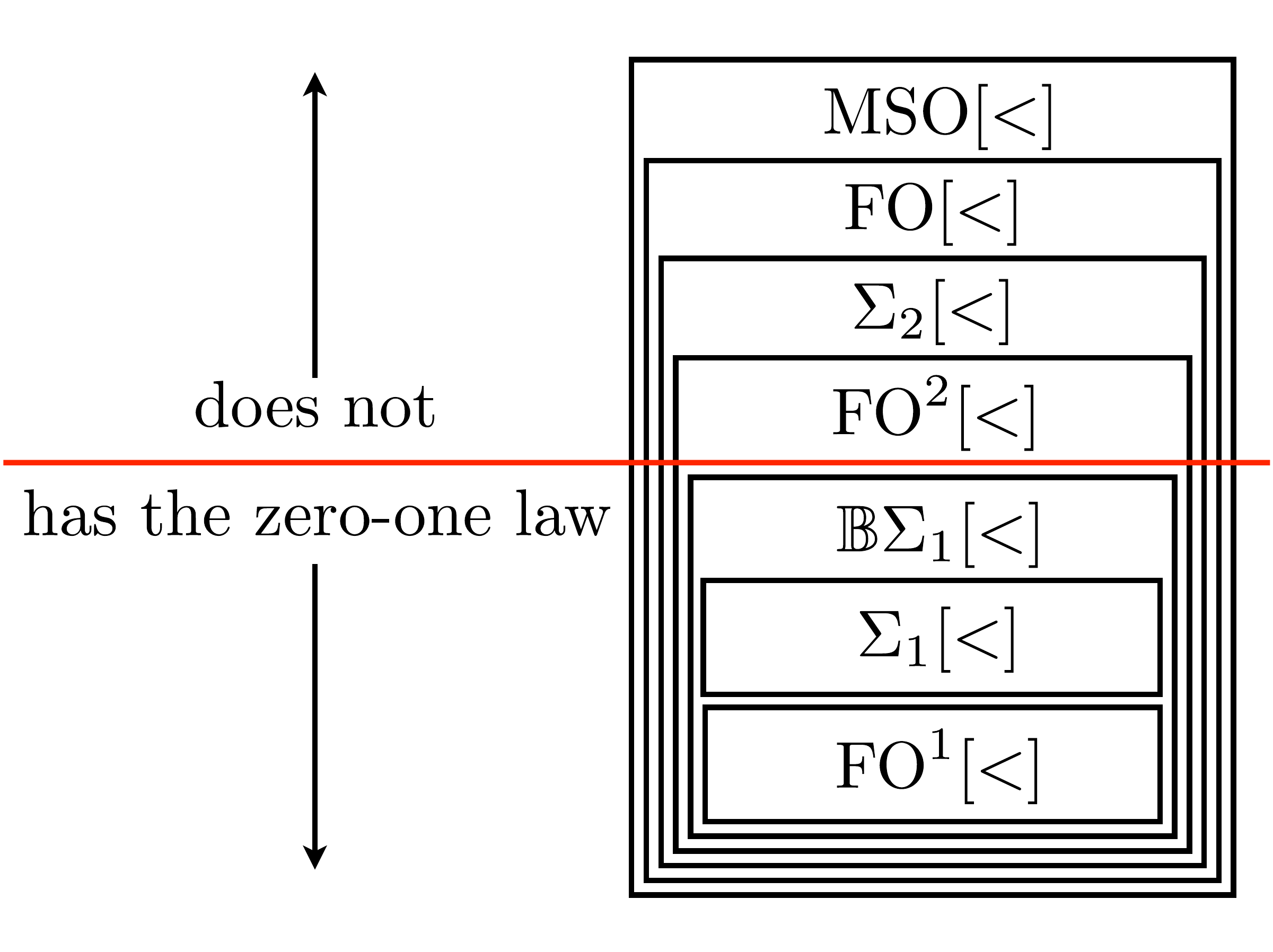}
\end{minipage}
 \caption{Logical and algebraic characterisations of well known
 subclasses of regular languages.}
 \label{fig:separation}
\end{figure}

The question then arises as to {\it which fragments of $\FO[<]$ over
finite words have the zero-one law}.
The algebraic characterisation of the zero-one law partially answers
this question.
Since every $\cal J$-{\it trivial} syntactic monoid has a zero element
(\cf \cite{MPRI}), Theorem \ref{thm:zero} leads to the following corollary.
\begin{corollary}\label{cor:logic}
 The Boolean closure of existential first-order logic over finite words
 has the zero-one law.
\end{corollary}
One can easily verify that the sentence $\Phi_{aA^*}$ in example
\ref{ex:logic}, which only uses two variables $i$ and $j$, is in
$\FO^2[<]$. It follows that $\FO^2[<]$ does not have the zero-one law,
hence Corollary \ref{cor:logic} shows us a ``separation line'' (red line
in Figure \ref{fig:separation}).
It must be noted that the class of zero-one languages $\ZO$ and unambiguous
polynomials are incomparable.
To take a simple example, consider two languages $(aa)^*$ and $aA^*$
over $A = \{a, b\}$.
The language $(aa)^*$ is zero-one but not unambiguous polynomial since
its syntactic monoid is not {\it aperiodic} (\ie having no nontrivial
subgroup).
Conversely, $aA^*$ is not zero-one but unambiguous polynomial since it
is definable in $\FO^2[<]$ as we have stated in Example \ref{ex:logic}.
An interesting open problem is whether there exists a logical fragment
that exactly captures the zero-one law.

%% file: chaps/conclusion.tex
\section{Related works}\label{conclusion}
The notion of probability $\mu_n$ for regular languages has been studied
by Berstel \cite{berstel:hal-00619884} from 1973, and by Salomaa and
Soittola \cite{Salomaa:1978:ATA:578607} from 1978 in the context of the
{\it theory of formal power series}. They proved that $\mu_n(L)$ has
finitely many accumulation points and each accumulation point is
rational.
Another approach, based on {\it Markov chain theory}, was presented by
Bodirsky {\it et al}. \cite{Bodirsky}. 
They investigate the algorithmic complexity of computing accumulation
points of $L$ and introduced an $\order(n^3)$ algorithm to compute
$\mu(L)$ for any regular language $L$ (and hence whether $L$ is
zero-one), if $L$ is given by an $n$-states deterministic finite
automaton.

A similar notion, {\it density} of a language have also been
 studied in {\it algebraic coding theory} (\cf
 \cite{opac-b1092233,Berstel:2009:CA:1708078}).
A {\it probability distribution} $\pi$ on $A^*$
is a function $\pi: A^* \rightarrow [0,1]$ such that $\pi(\epsilon) = 1$
 and $\sum_{a \in A} \pi(wa) = \pi(w)$ for all $w$ in $A^*$. As a
 particular case, a {\it Bernoulli distribution} is a morphism from
 $A^*$ into $[0,1]$ such that $\sum_{a \in A} \pi(a) = 1$. Clearly, a
 Bernoulli distribution is a probability distribution.
We denote by $A^{(n)} = A^0 \cup A \cup \cdots \cup A^{n-1}$ the set of
all words of length less than $n$ over a finite alphabet $A$.
The {\it density} $\delta(L)$ of $L$ is a limit defined by
\[
 \delta(L) = \lim_{n \rightarrow \infty} \frac{1}{n} \pi\left(L \cap A^{(n)}\right)
\]
where $\pi$ is a probability distribution on $A^*$. 
A monoid $M$ is called {\it well founded} if it has a unique minimal
ideal, if moreover this ideal is the union of the minimal left ideals of
$M$, and also of the minimal right ideals, and if the intersection of a
minimal right ideal and of a minimal left ideal is a finite group.
An elementary result from analysis shows that if the sequence $\pi(L
\cap A^n)$ has a limit, then $\delta(L)$ also has a limit, and both are
equal. The converse, however, does not hold (\eg $\delta((AA)^*) = 1/2$).
In their book \cite{Berstel:2009:CA:1708078}, Berstel {\it et al}. proved
Theorem 13.4.5 which states that, for any well founded monoid $M$ and
morphism $\phi: A^* \rightarrow M$, $\delta(\phi^{-1}(m))$ has a
limit for every $m$ in $M$. Furthermore, this density is non-zero if and
only if $m$ in the minimal ideal $K$ of $M$ from which we obtain
$\delta(\phi^{-1}(K)) = 1$.
Since every monoid with zero is well founded, Theorem 13.4.5 implies
that, every language with zero is zero-one (\ie $\con{zero} \Rightarrow
\con{zeroone}$, ``easy part'' of our Theorem \ref{thm:zero}).
Some other related results can be found in the {\it theory of
probabilities on algebraic structures} initiated by Grenander
\cite{opac-b1108561} and Martin-L\"of \cite{Martin-Lof:1965:PTD}.

The point to observe is that the techniques presented in this paper are
purely automata theoretic. We did not use any probability theoretic
tools, like as measure theory, formal power series, Markov chain, algebraic coding
theory, {\it etc}. This point deserves explicit emphasise.\\

\sect{Acknowledgement}
I wish to thank the anonymous reviewers for their valuable comments and
suggestions to improve the quality of the paper, especially, who
informed me the previous works in algebraic coding theory (Theorem
13.4.5 in \cite{Berstel:2009:CA:1708078}).
Special thanks also go to Prof. Yasuhiko Minamide (Tokyo Institute of
Technology) whose meticulous comments for Lemma \ref{lemma:finalstates}
were an enormous help to me.
I am grateful to Prof. Jacques Sakarovitch (T\'el\'ecom ParisTech)
whose comments and suggestions (and his excellent book \cite{Sakarovitch:2009:EAT:1629683}) were
innumerably valuable throughout the course of my study.
This work was supported by JSPS KAKENHI Grant Number $26 \cdot 11962$.